\documentclass[12pt]{article}
\usepackage{amsmath,amsthm}
\usepackage{eufrak}
\DeclareMathAlphabet{\bb}{U}{msb}{m}{n} \gdef\C{\bb C} \gdef\dZ{\bb
Z}   \gdef\dS{\bb S} \gdef\R{\bb R}
\gdef\K{\bb K} \gdef\BH{\bb H} \gdef\F{\bb F} 

 \DeclareMathOperator{\spin}{{\bf
Spin}}

\DeclareMathOperator{\Sym}{Sym}

\DeclareMathOperator{\Tr}{Tr} \DeclareMathOperator{\SL}{SL}
\DeclareMathOperator{\SO}{SO}\DeclareMathOperator{\SU}{SU}

\newcommand{\cA}{\mathcal{A}}

\newcommand{\bcE}{\boldsymbol{\mathcal{E}}}
\newcommand{\bcK}{\boldsymbol{\mathcal{K}}}

\newcommand{\cP}{{\cal P}}

\newcommand{\sH}{{\sf H}}

\newcommand{\sX}{{\sf X}}
\newcommand{\sY}{{\sf Y}}

\newcommand{\bsH}{{\boldsymbol{\sf H}}}

\newcommand{\fA}{\mathfrak{A}}

\newcommand{\fP}{\mathfrak{P}}

\newcommand{\fS}{\mathfrak{S}}
\newcommand{\fU}{\mathfrak{U}}

\newcommand{\cl}{C\kern -0.2em \ell}

\newtheorem{thm}{Theorem}
\begin{document}
\title{On algebraic structure of matter spectrum}
\author{V.~V. Varlamov\thanks{Siberian State Industrial University,
Kirova 42, Novokuznetsk 654007, Russia, e-mail:
varlamov@sibsiu.ru}}
\date{}
\maketitle
\begin{abstract}
An algebraic structure of matter spectrum is studied. It is shown that a base mathematical construction, lying in the ground of matter spectrum (introduced by Heisenberg) , is a two-level Hilbert space. Two-level structure of the Hilbert space is defined by the following pair: 1) a separable Hilbert space with operator algebras and fundamental symmetries; 2) a nonseparable (physical) Hilbert space with dynamical and gauge symmetries, that is, a space of states (energy levels) of the matter spectrum. The each state of matter spectrum is defined by a cyclic representation within Gel'fand-Naimark-Segal construction. A decomposition of the physical Hilbert space onto coherent subspaces is given. This decomposition allows one to describe the all observed spectrum of states on an equal footing, including lepton, meson and baryon sectors of the matter spectrum. Following to Heisenberg, we assume that there exist no fundamental particles, there exist fundamental symmetries. It is shown also that all the symmetries of matter spectrum are divided onto three kinds: fundamental, dynamical and gauge symmetries.
\end{abstract}
\section{Introduction}
As is known, one of the most important undecided problems in theoretical physics is a description of mass spectrum of elementary particles (it is one from 30 problems in Ginzburg's list \cite{Ginz}). A some progress in systematization of hadron spectra has been achieved in $\SU(3)$- and $\SU(6)$-theories (see, for example, \cite{RF70}). In the recent past it has been attempted to describe a mass distribution of baryon octets within quark model and its extensions \cite{Guz,Ples}. However, a lepton sector of the particle spectrum takes no place in the framework of dynamical symmetries ($\SU(N)$-theories). Furthermore, a gauge sector has an isolated position. Such a triple division of the particle spectrum is a distinctive feature of the standard model (SM), in which an existence of the three kinds of `fundamental particles' (quarks, leptons and gauge bosons) is postulated. It is not hard to see that in SM we have a description scheme of the particle spectrum from a position of reductionism, according to which all the members of hadron sector (baryons and mesons) are constructed from the quarks, and leptons and gauge bosons are understood as fundamental particles.

As it is known, an antithesis to reductionism is a holism (teaching about whole). A description scheme of particle spectrum from a position of holism (an alternative to reductionism of SM) is the Heisenberg's approach \cite{Heisen1,Heisen}. According to Heisenberg, in the ground of the all wide variety of elementary particles we have a certain \textit{substrate of energy}, the outline form of which (via the fundamental symmetries) is a \textbf{\textit{matter spectrum}}. The each level (\textit{state}) of matter spectrum is defined by a representation of a group of fundamental symmetry. The each elementary particle presents itself a some energy level of this spectrum. An essential distinctive feature of such description is an absence of fundamental particles. Heisenberg claimed that a notion `consists of' (the main notion of reductionism) does not work in particle physics. Applying this notion, we obtain that the each particle consists of the all known particles. For that reason among the all elementary particles we cannot to separate one as a fundamental particle \cite{Heisen}. In Heisenberg's approach we have fundamental symmetries instead fundamental particles. In Heisenberg's opinion, all known symmetries in particle physics are divided on the two categories: \textit{fundamental (primary) symmetries} (such as the Lorentz group, discrete symmetries, conformal group) and \textit{dynamical (secondary) symmetries} (such as $\SU(3)$, $\SU(6)$ and so on).

In the present paper we study an algebraic formulation of the main notions of matter spectrum. A general algebraic structure of matter spectrum, defined by a two-level Hilbert space, is given in the section 2. It is shown that basic energy levels (states) of matter spectrum are constructed on the ground of cyclic representations within Gel'fand-Naimark-Segal construction. A concrete realization of the operator algebra is realized via the spinor structure associated with the each cyclic representation. Pure states (cyclic representations), defining the levels of matter spectrum, are divided with respect to a charge (an action of a pseudoautomorphism of the spinor structure) on the subsets of charged, neutral and truly neutral states. It is shown that the structure of matter spectrum is defined by a partition of physical Hilbert space (a space of states) onto coherent subspaces. At this point, a superposition principle takes place in the restricted form, that is, within the limits of coherent subspaces. None of the levels (states) of matter spectrum is separate or `fundamental', all the levels present actualized (localized) states of quantum micro-objects. Notion of symmetry plays a key role. Symmetries of matter spectrum are divided on the three kinds: \textit{fundamental} symmetries which participate in formation of states, \textit{dynamical} and \textit{gauge} symmetries which relate states with each other.
\section{General algebraic structure of matter spectrum}
An algebraic formulation of quantum theory was first proposed by von Neumann \cite{Neu36} and (in language of $C^\ast$-algebras) by Segal \cite{Seg47}. Further, an algebraic formulation of local quantum field theory was analyzed in detail by Haag \cite{Haa59} and by Araki \cite{Ara61} (see also \cite{HS62,HK64}). In this section we consider a general structure (algebraic formulation) of matter spectrum, which is defined by a construction of the Hilbert space of elementary particle. In its turn, this Hilbert space has a two-level structure. On the first level we have a \textit{separable Hilbert space} $\sH_\infty$. In $\sH_\infty$ according to standard rules of local quantum phenomenology (Schr\"{o}dinger picture) we have \textbf{\textit{observables}} ($C^\ast$-algebras), states, spectra of observables and \textit{fundamental symmetries}. The basic observable is the \textit{energy} (Hermitian operator  $H$), the fundamental symmetry is defined by the Lorentz group $\SO_0(1,3)$. On the second level we have a \textit{nonseparable Hilbert space} $\bsH^S\otimes\bsH^Q\otimes\bsH_\infty$, in which the main structural forming components are the \textbf{\textit{states}}\footnote{Segal \cite{Segal} pointed out that the fundamental object associated with a physical system may be taken either as an \textit{observable} or as a \textit{state}. Segal wrote: ``Whether observables or states are more fundamental is somewhat parallel to the same question for chickens and eggs. Leaving aside metaphysics, either notion has certain distinctive advantages as a foundational concept, but no analytical treatment starting from the states exists as yet which is of the same order of comprehensiveness and applicability as that starting from observables. In particular, the work of Birkhoff and von Neumann (1936) and of Mackey (1957), in which the states play the fundamental role, has not yet been developed to the point where their serviceability as possible frameworks for quantum field phenomenology is apparent'' \cite[p.~13]{Segal}.} (rays). State vectors in $\bsH^S\otimes\bsH^Q\otimes\bsH_\infty$ are constructed from irreducible finite-dimensional representations of the group $\SL(2,\C)$. These representations are defined in eigenvector subspaces $\sH_E\subset\sH_\infty$ of the energy operator $H$. Thus, state vectors in $\bsH^S\otimes\bsH^Q\otimes\bsH_\infty$ define spin and charge degrees of freedom of elementary particle. An elementary particle presents itself a superposition of state vectors in $\bsH^S\otimes\bsH^Q\otimes\bsH_\infty$, that is, in the case of pure entangled states we have nonseparable (nonlocal) state. Therefore, two-level structure of the Hilbert space of elementary particle is defined by the pair $(\sH_\infty,\,\bsH^S\otimes\bsH^Q\otimes\bsH_\infty)$.
\subsection{Hilbert space $\sH_\infty$}
So, on the first level we have separable Hilbert space $\sH_\infty$, that is, $\sH_\infty$ is a Banach space endowed with enumerable base which is dense everywhere in $\sH_\infty$ (any element from $\sH_\infty$ is represented as a limit of sequence of the elements from enumerable set). Observables play a key role on the level $\sH_\infty$. An original object of consideration is $C^\ast$-algebra $\fA$ with the unit (\textit{algebra of observables} or algebra of bounded observables). Hermitian elements of this algebra are bounded \textit{observables}\footnote{Hermitian elements of $C^\ast$-algebra $\fA$ form \textit{Jordan algebra} $\fA_h$. In $\fA_h$ we have linear combinations with the real coefficients. The square of the each element in $\fA_h$ is defined by a symmetrical product (pseudoproduct) $A\circ B=1/4[(A+B)^2-(A-B)^2]$.}. A positive functional $\omega$ over $\fA$ (with the norm $\|\omega\|\equiv\omega(1)=1$) is called a \textit{state} of the algebra $\fA$. A set of the all states of $\fA$ we will denote as $S(\fA)$. The magnitude $\omega(A)$ at $A=A^\ast$ is understood as an \textit{average value} of the observable $A$ in the state $\omega$. $S(\fA)$ is a \textit{convex set}, that is, for any two states $\omega_1$, $\omega_2$ and $\lambda_1,\,\lambda_2\geq 0$, $\lambda_1+\lambda_2=1$, we have $\lambda_1\omega_1+\lambda_2\omega_2\in S(\fA)$. The state $\omega$ is called \textit{mixed state} (or statistical mixture), when $\omega$ can be represented in the form $\omega=\lambda\omega_1+(1-\lambda)\omega_2$, where $0<\lambda<1$ and $\omega_1$, $\omega_2$ are two different states of the algebra $\fA$. States, which cannot be represented as mixed states, are called \textit{pure states} (pure states are extremal points of the set $S(\fA)$). A set of the all pure states of $C^\ast$-algebra $\fA$ we denote via $PS(\fA)$. Let $\fS$ be a set of such states of the algebra $\fA$, for which the condition $\omega(A)\geq 0$ is fulfilled for the all $\omega\in\fS$, that is, $A$ is a positive element of the algebra $\fA$ ($A$ can be represented in the form $A=B^\ast B$). Then $\fS$ is called a set of \textit{physical states} of $\fA$, and the pair $(\fA,\fS)$ is called a \textit{physical system}.

For the arbitrary $C^\ast$-algebra $\fA$ a transition probability between two pure states $\omega_1,\,\omega_2\in PS(\fA)$ is given by the Roberts-Roepstorff formula \cite{RR69}:
\[
|\langle\Phi_1\mid\Phi_2\rangle|^2=\omega_1\cdot\omega_2=1-1/4\|\omega_1-\omega_2\|^2,
\]
where $\left|\Phi_1\right\rangle$ and $\left|\Phi_2\right\rangle$ are unit vectors of the space $\sH_\infty$. At this point, $\omega_1\cdot\omega_2=\omega_2\cdot\omega_1$ and $\omega_1\cdot\omega_2$ always belongs to $[0,1]$. Correspondingly, $\omega_1\cdot\omega_2=1$ exactly when $\omega_1=\omega_2$. Two pure states $\omega_1$ and $\omega_2$ are called \textit{orthogonal states} if the transition probability $\omega_1\cdot\omega_2$ is equal to zero. Therefore, two subsets $S_1$ and $S_2$ in $PS(\fA)$ are mutually orthogonal when $\omega_1\cdot\omega_2=0$ for the all $\omega_1\in S_1$ and $\omega_2\in S_2$. Further, nonempty subset $S\in PS(\fA)$ is called \textit{indecomposable set} in the case when $S$ cannot be divided on the two orthogonal subsets. Following to Haag and Kastler \cite{HK64}, we assume that any maximal indecomposable set is a \textit{sector}. So, $PS(\fA)$ is divided on the sectors, therefore, in $PS(\fA)$ there exists an equivalence relation $\omega_1\sim\omega_2$ if and only if there exists an indecomposable set in $PS(\fA)$ containing $\omega_1$ and $\omega_2$. Therefore, $PS(\fA)$ is divided in pairs on disjoint and mutually orthogonal sectors which coincide with \textit{equivalence classes} in $PS(\fA)$.

One of the most important aspects in theory of $C^\ast$-algebras is a duality between states and representations. A relation between states and irreducible representations of operator algebras was first formulated by Segal \cite{Se47}. Let $\pi$ be a some representation of the algebra $\fA$ in the Hilbert space $\sH_\infty$, then for any non-null vector $\left|\Phi\right\rangle\in\sH_\infty$ the expression
\begin{equation}\label{VectState}
\omega_\Phi(A)=\frac{\langle\Phi\mid\pi(A)\Phi\rangle}{\langle\Phi\mid\Phi\rangle}
\end{equation}
defines a state $\omega_\Phi(A)$ of the algebra $\fA$. $\omega_\Phi(A)$ is called a \textit{vector state} associated with the representation $\pi$ ($\omega_\Phi(A)$ corresponds to the vector $\left|\Phi\right\rangle$). Let $\rho$ be a \textit{density matrix} in $\sH_\infty$, then
\[
\omega_\rho(A)=\Tr\left(\rho\pi(A)\right).
\]
Analogously, $\omega_\rho(A)$ is a state associated with $\pi$ and $\omega_\rho(A)$ corresponds to density matrix $\rho$. The states $\omega_\rho(A)$ are statistical mixtures of the vector states (\ref{VectState}). Let $S_\pi$ be a set of all states associated with the representation $\pi$. Two representations $\pi_1$ and $\pi_2$ with one and the same set of associated states (that is, $S_{\pi_1}=S_{\pi_2}$) are called phenomenologically equivalent sets (it corresponds to unitary equivalent representations). Moreover, the set $PS(\fA)$ of the all pure states of $C^\ast$-algebra $\fA$ coincides with the set of all vector states associated with the all \textit{irreducible representations} of the algebra $\fA$.

Further, let $\pi$ be a representation of $C^\ast$-algebra $\fA$ in $\sH_\infty$ and let $\left|\Phi\right\rangle$ be a \textit{cyclic vector}\footnote{Vector $\left|\Phi\right\rangle\in\sH_\infty$ is called a cyclic vector for the representation $\pi$, if the all vectors $\left|\pi(A)\Phi\right\rangle$ (where $A\in\fA$) form a total set in $\sH_\infty$, that is, such a set, for which a closing of linear envelope is dense everywhere in $\sH_\infty$. $\pi$ with the cyclic vector is called a cyclic representation.} of the representation $\pi$ defining the state $\omega_\Phi$. In accordance with Gel'fand-Naimark-Segal construction (see \cite{BLOT}) the each state defines a some representation of the algebra $\fA$. At this point, resulting representation is irreducible exactly when the state is pure. Close relationship between states and representations of $C^\ast$-algebra, based on the GNS construction, allows us to consider representations of the algebra as an effective tool for organization of states. This fact becomes more evident at the concrete realization $\pi(\fA)$.
\subsubsection{Fundamental symmetries}
Usually (in abstract-algebraic formulation), symmetry is understood as a transformation of physical system, which does not change its structural properties. In its turn, physical system is characterized by the algebra of observables  $\fA$ and by the set of states $S(\fA)$. Following to Heisenberg, we assume that on the level of separable Hilbert space $\sH_\infty$ we have \textit{fundamental (primary) symmetries}. Let us define fundamental symmetry as a pair of bijections $\alpha:\;\fA\rightarrow\fA$ and $\alpha^\prime:\;S(\fA)\rightarrow S(\fA)$, satisfying to coordination condition: $(\alpha^\prime\omega)(\alpha A)=\omega(A)$ for the all $A\in\fA$, $\omega\in S(\fA)$. A set of all symmetries of the physical system form a group with multiplication defined  by a composition of bijections. The product of the two symmetries $(\alpha,\alpha^\prime)$ and $(\beta,\beta^\prime)$ is a symmetry $(\alpha\beta,\alpha^\prime\beta^\prime)$, where $(\alpha\beta)(A)\equiv\alpha[\beta(A)]$ and $(\alpha^\prime\beta^\prime)(\omega)\equiv\alpha^\prime[\beta^\prime(\omega)]$.

The group $G$ is called a \textit{group of fundamental symmetry} of $C^\ast$-algebra $\fA$ when there exists a homomorphism $g\rightarrow(\alpha_g,\alpha^\prime_g)$ of the group $G$ into the group of all symmetries of the system $(\fA,S(\fA))$. We assume that $G$ is a noncompact Lie group (for example, Lorentz group, Poincar\'{e} group or conformal group). Then the following continuity condition is fulfilled: at any physical state $\omega\in\fS$ and any fixed $A\in\fA$ the function $g\rightarrow\omega(\alpha_g(A))$ is continuous on $g$. At this point, the group $G$ is unitary-antiunitary realized if there exists a continuous representation $g\rightarrow U_g$ of the group $G$ defined by unitary or antiunitary operators (that is, $\alpha_g$ are algebraic automorphisms or antiautomorphisms) in the Hilbert space $\sH_\infty$ such that for the all $A\in\fA$, $g\in G$ there is $\alpha_g(A)=U_gA^{(\ast)}U^{-1}_g$, where $A^{(\ast)}$ is $A$ for unitary $U_g$ and $A^\ast$ for antiunitary $U_g$.
\subsubsection{Concrete realization $\pi(\fA)$}
In this section we will consider a concrete realization of the operator algebra $\fA$. A transition $\fA\Rightarrow\pi(\fA)$ from $\fA$ to a concrete algebra $\pi(\fA)$ is called sometimes as `clothing'. So, the basic observable is \textit{energy} which represented by Hermitian operator $H$. Let $G=\SO_0(1,3)\simeq\SL(2,\C)/\dZ_2$ be the group of fundamental symmetry, where $\SO_0(1,3)$ is the Lorentz group. Let $\widetilde{G}\simeq\SL(2,\C)$ be the \textit{universal covering} of $\SO_0(1,3)$. Let $H$ be the energy operator defined on the separable Hilbert space $\sH_\infty$. Then all the possible values of energy (states) are eigenvalues of the operator $H$. At this point, if $E_1\neq E_2$ are eigenvalues of $H$, and $\left|\Phi_1\right\rangle$ and $\left|\Phi_2\right\rangle$ are corresponding eigenvectors in the space $\sH_\infty$, then $\langle\Phi_1\mid\Phi_2\rangle=0$. All the eigenvectors, belonging to a given eigenvalue $E$, form (together with the null vector) an \textit{eigenvector subspace} $\sH_E$ of the Hilbert space $\sH_\infty$. All the eigenvector subspaces $\sH_E\in\sH_\infty$ are finite-dimensional. A dimensionality $r$ of $\sH_E$ is called a \textit{multiplicity} of the eigenvalue $E$. When $r>1$ the eigenvalue $E$ is \textit{$r$-fold degenerate}. Further, let $\sX_l$, $\sY_l$ be infinitesimal operators of the complex envelope of the group algebra $\mathfrak{sl}(2,\C)$ for universal covering $\widetilde{G}$, $l=1,2,3$. As is known \cite{BHJ26}, the energy operator $H$ commutes with the all operators in $\sH_\infty$, which represent a Lie algebra of the group $\widetilde{G}$. Let us consider an arbitrary eigenvector subspace $\sH_E$ of the energy operator $H$. Since the operators $\sX_l$, $\sY_l$ and $H$ commute with the each other, then, as is known \cite{Dir}, for these operators we can build a common system of eigenfunctions. It means that the subspace $\sH_E$ is invariant with respect to operators $\sX_l$, $\sY_l$ (moreover, the operators $\sX_l$, $\sY_l$ can be considered \textit{only on} $\sH_E$). Further, we suppose that there is a some \textit{local representation} of the group $\widetilde{G}$ defined by the operators acting in the space $\sH_\infty$. At this point, we assume that all the representing operators commute with $H$. Then the each eigenvector subspace $\sH_E$ of the energy operator is invariant with respect to operators of complex momentum $\sX_l$, $\sY_l$. It allows us to identify subspaces $\sH_E$ with symmetrical spaces $\Sym_{(k,r)}$ of interlocking representations $\boldsymbol{\tau}_{k/2,r/2}$ of the Lorentz group. Thus, we obtain a concrete realization (`clothing') of the operator algebra $\pi(\fA)\rightarrow\pi(H)$, where $\pi\equiv\boldsymbol{\tau}_{k/2,r/2}$. The system of interlocking representations of the Lorentz group is shown on the Fig.\,1 (for more details see \cite{Var03,Var07}). Hence it follows that the each possible value of energy (energy level) is a vector state of the form (\ref{VectState}):
\begin{equation}\label{VectState2}
\omega_\Phi(H)=\frac{\langle\Phi\mid\pi(H)\Phi\rangle}{\langle\Phi\mid\Phi\rangle}=
\frac{\langle\Phi\mid\boldsymbol{\tau}_{k/2,r/2}(H)\Phi\rangle}{\langle\Phi\mid\Phi\rangle},
\end{equation}
The state $\omega_\Phi(H)$ is associated with the representation $\pi\equiv\boldsymbol{\tau}_{k/2,r/2}$ and the each $\omega_\Phi(H)$ corresponds to non-null (cyclic) vector $\left|\Phi\right\rangle\in\sH_\infty$. Analogously, if $\rho$ is the density matrix in $\sH_\infty$, then
\[
\omega_\rho(H)=\Tr\left(\rho\boldsymbol{\tau}_{k/2,r/2}(H)\right)
\]
is a statistical mixture of the vector states (\ref{VectState2}).
\begin{figure}[ht]
\unitlength=1.5mm
\begin{center}
\begin{picture}(100,50)
\put(50,0){$\overset{(0,0)}{\bullet}$}\put(47,5.5){\line(1,0){10}}\put(52.25,2.75){\line(0,1){7.25}}
%\put(48,3){$\scr(0,0)$}
\put(55,5){$\overset{(\frac{1}{2},0)}{\bullet}$}
%\put(53,8){$\scr(\tfrac{1}{2},0)$}
\put(45,5){$\overset{(0,\frac{1}{2})}{\bullet}$}
\put(40,10){$\overset{(0,1)}{\bullet}$}\put(42,10.5){\line(1,0){10}}\put(47.25,7.75){\line(0,1){7.25}}
\put(50,10){$\overset{(\frac{1}{2},\frac{1}{2})}{\bullet}$}
\put(52,10.5){\line(1,0){10}}\put(57.25,7.75){\line(0,1){7.25}}
\put(60,10){$\overset{(1,0)}{\bullet}$}
\put(35,15){$\overset{(0,\frac{3}{2})}{\bullet}$}\put(37,15.5){\line(1,0){10}}\put(42.25,12.75){\line(0,1){7.25}}
\put(45,15){$\overset{(\frac{1}{2},1)}{\bullet}$}\put(47,15.5){\line(1,0){10}}\put(52.25,12.75){\line(0,1){7.25}}
\put(55,15){$\overset{(1,\frac{1}{2})}{\bullet}$}\put(57,15.5){\line(1,0){10}}\put(62.25,12.75){\line(0,1){7.25}}
\put(65,15){$\overset{(\frac{3}{2},0)}{\bullet}$}
\put(30,20){$\overset{(0,2)}{\bullet}$}\put(32,20.5){\line(1,0){10}}\put(37.25,17.75){\line(0,1){7.25}}
\put(40,20){$\overset{(\frac{1}{2},\frac{3}{2})}{\bullet}$}
\put(42,20.5){\line(1,0){10}}\put(47.25,17.75){\line(0,1){7.25}}
\put(50,20){$\overset{(1,1)}{\bullet}$}\put(52,20.5){\line(1,0){10}}\put(57.25,17.75){\line(0,1){7.25}}
\put(60,20){$\overset{(\frac{3}{2},\frac{1}{2})}{\bullet}$}
\put(62,20.5){\line(1,0){10}}\put(67.25,17.75){\line(0,1){7.25}}
\put(70,20){$\overset{(2,0)}{\bullet}$}
\put(25,25){$\overset{(0,\frac{5}{2})}{\bullet}$}\put(27,25.5){\line(1,0){10}}\put(32.25,22.75){\line(0,1){7.25}}
\put(35,25){$\overset{(\frac{1}{2},2)}{\bullet}$}\put(37,25.5){\line(1,0){10}}\put(42.25,22.75){\line(0,1){7.25}}
\put(45,25){$\overset{(1,\frac{3}{2})}{\bullet}$}\put(47,25.5){\line(1,0){10}}\put(52.25,22.75){\line(0,1){7.25}}
\put(55,25){$\overset{(\frac{3}{2},1)}{\bullet}$}\put(57,25.5){\line(1,0){10}}\put(62.25,22.75){\line(0,1){7.25}}
\put(65,25){$\overset{(2,\frac{1}{2})}{\bullet}$}\put(67,25.5){\line(1,0){10}}\put(72.25,22.75){\line(0,1){7.25}}
\put(75,25){$\overset{(\frac{5}{2},0)}{\bullet}$}
\put(20,30){$\overset{(0,3)}{\bullet}$}\put(22,30.5){\line(1,0){10}}\put(27.25,27.75){\line(0,1){7.25}}
\put(30,30){$\overset{(\frac{1}{2},\frac{5}{2})}{\bullet}$}
\put(32,30.5){\line(1,0){10}}\put(37.25,27.75){\line(0,1){7.25}}
\put(40,30){$\overset{(1,2)}{\bullet}$}\put(42,30.5){\line(1,0){10}}\put(47.25,27.75){\line(0,1){7.25}}
\put(50,30){$\overset{(\frac{3}{2},\frac{3}{2})}{\bullet}$}
\put(52,30.5){\line(1,0){10}}\put(57.25,27.75){\line(0,1){7.25}}
\put(60,30){$\overset{(2,1)}{\bullet}$}\put(62,30.5){\line(1,0){10}}\put(67.25,27.75){\line(0,1){7.25}}
\put(70,30){$\overset{(\frac{5}{2},\frac{5}{2})}{\bullet}$}
\put(72,30.5){\line(1,0){10}}\put(77.25,27.75){\line(0,1){7.25}}
\put(80,30){$\overset{(3,0)}{\bullet}$}
\put(15,35){$\overset{(0,\frac{7}{2})}{\bullet}$}\put(17,35.5){\line(1,0){10}}\put(22.25,32.75){\line(0,1){7.25}}
\put(25,35){$\overset{(\frac{1}{2},3)}{\bullet}$}\put(27,35.5){\line(1,0){10}}\put(32.25,32.75){\line(0,1){7.25}}
\put(35,35){$\overset{(1,\frac{5}{2})}{\bullet}$}\put(37,35.5){\line(1,0){10}}\put(42.25,32.75){\line(0,1){7.25}}
\put(45,35){$\overset{(\frac{3}{2},2)}{\bullet}$}\put(47,35.5){\line(1,0){10}}\put(52.25,32.75){\line(0,1){7.25}}
\put(55,35){$\overset{(2,\frac{3}{2})}{\bullet}$}\put(57,35.5){\line(1,0){10}}\put(62.25,32.75){\line(0,1){7.25}}
\put(65,35){$\overset{(\frac{5}{2},1)}{\bullet}$}\put(67,35.5){\line(1,0){10}}\put(72.25,32.75){\line(0,1){7.25}}
\put(75,35){$\overset{(3,\frac{1}{2})}{\bullet}$}\put(77,35.5){\line(1,0){10}}\put(82.25,32.75){\line(0,1){7.25}}
\put(85,35){$\overset{(\frac{7}{2},0)}{\bullet}$}
\put(10,40){$\overset{(0,4)}{\bullet}$}\put(12,40.5){\line(1,0){10}}
\put(20,40){$\overset{(\frac{1}{2},\frac{7}{2})}{\bullet}$}\put(22,40.5){\line(1,0){10}}
\put(30,40){$\overset{(1,3)}{\bullet}$}\put(32,40.5){\line(1,0){10}}
\put(40,40){$\overset{(\frac{3}{2},\frac{5}{2})}{\bullet}$}\put(42,40.5){\line(1,0){10}}
\put(50,40){$\overset{(2,2)}{\bullet}$}\put(52,40.5){\line(1,0){10}}
\put(60,40){$\overset{(\frac{5}{2},\frac{3}{2})}{\bullet}$}\put(62,40.5){\line(1,0){10}}
\put(70,40){$\overset{(3,1)}{\bullet}$}\put(72,40.5){\line(1,0){10}}
\put(80,40){$\overset{(\frac{7}{2},\frac{1}{2})}{\bullet}$}\put(82,40.5){\line(1,0){10}}
\put(90,40){$\overset{(4,0)}{\bullet}$}
\put(11.5,45){$\vdots$}
\put(21.5,45){$\vdots$}
\put(31.5,45){$\vdots$}
\put(41.5,45){$\vdots$}
\put(51.5,45){$\vdots$}
\put(61.5,45){$\vdots$}
\put(71.5,45){$\vdots$}
\put(81.5,45){$\vdots$}
\put(91.5,45){$\vdots$}
\put(10,0.5){\line(1,0){42}}\put(50,0.5){\vector(1,0){42}}
\put(16.5,32){$\vdots$}
\put(16.5,29){$\vdots$}
\put(16.5,26){$\vdots$}
\put(16.5,23){$\vdots$}
\put(16.5,20){$\vdots$}
\put(16.5,17){$\vdots$}
\put(16.5,14){$\vdots$}
\put(16.5,11){$\vdots$}
\put(16.5,9){$\vdots$}
\put(16.5,6){$\vdots$}
\put(16.5,3){$\vdots$}
\put(16.5,1.5){$\cdot$}
\put(16.5,0){$\cdot$}
\put(14.5,-3){$-\frac{7}{2}$}
\put(21.5,27){$\vdots$}
\put(21.5,24){$\vdots$}
\put(21.5,21){$\vdots$}
\put(21.5,18){$\vdots$}
\put(21.5,15){$\vdots$}
\put(21.5,13){$\vdots$}
\put(21.5,9){$\vdots$}
\put(21.5,6){$\vdots$}
\put(21.5,3){$\vdots$}
\put(21.5,1.5){$\cdot$}
\put(21.5,0){$\cdot$}
\put(19.5,-3){$-3$}
\put(26.5,22){$\vdots$}
\put(26.5,19){$\vdots$}
\put(26.5,16){$\vdots$}
\put(26.5,13){$\vdots$}
\put(26.5,10){$\vdots$}
\put(26.5,7){$\vdots$}
\put(26.5,4){$\vdots$}
\put(26.5,1){$\vdots$}
\put(24.5,-3){$-\frac{5}{2}$}
\put(31.5,17){$\vdots$}
\put(31.5,14){$\vdots$}
\put(31.5,11){$\vdots$}
\put(31.5,8){$\vdots$}
\put(31.5,5){$\vdots$}
\put(31.5,2){$\vdots$}
\put(31.5,0.5){$\cdot$}
\put(29.5,-3){$-2$}
\put(36.5,12){$\vdots$}
\put(36.5,9){$\vdots$}
\put(36.5,6){$\vdots$}
\put(36.5,3){$\vdots$}
\put(36.5,1.5){$\cdot$}
\put(36.5,0){$\cdot$}
\put(34.5,-3){$-\frac{3}{2}$}
\put(41.5,7){$\vdots$}
\put(41.5,4){$\vdots$}
\put(41.5,1){$\vdots$}
\put(39.5,-3){$-1$}
\put(46.5,2){$\vdots$}
\put(46.5,0.5){$\cdot$}
\put(44.5,-3){$-\frac{1}{2}$}
\put(51.5,-3){$0$}
\put(56.5,2){$\vdots$}
\put(56.5,0.5){$\cdot$}
\put(56.5,-3){$\frac{1}{2}$}
\put(61.5,7){$\vdots$}
\put(61.5,4){$\vdots$}
\put(61.5,1){$\vdots$}
\put(61.5,-3){$1$}
\put(66.5,12){$\vdots$}
\put(66.5,9){$\vdots$}
\put(66.5,6){$\vdots$}
\put(66.5,3){$\vdots$}
\put(66.5,1.5){$\cdot$}
\put(66.5,0){$\cdot$}
\put(66.5,-3){$\frac{3}{2}$}
\put(71.5,17){$\vdots$}
\put(71.5,14){$\vdots$}
\put(71.5,11){$\vdots$}
\put(71.5,8){$\vdots$}
\put(71.5,5){$\vdots$}
\put(71.5,2){$\vdots$}
\put(71.5,0.5){$\cdot$}
\put(71.5,-3){$2$}
\put(76.5,22){$\vdots$}
\put(76.5,19){$\vdots$}
\put(76.5,16){$\vdots$}
\put(76.5,13){$\vdots$}
\put(76.5,10){$\vdots$}
\put(76.5,7){$\vdots$}
\put(76.5,4){$\vdots$}
\put(76.5,1){$\vdots$}
\put(76.5,-3){$\frac{5}{2}$}
\put(81.5,27){$\vdots$}
\put(81.5,24){$\vdots$}
\put(81.5,21){$\vdots$}
\put(81.5,18){$\vdots$}
\put(81.5,15){$\vdots$}
\put(81.5,13){$\vdots$}
\put(81.5,9){$\vdots$}
\put(81.5,6){$\vdots$}
\put(81.5,3){$\vdots$}
\put(81.5,1.5){$\cdot$}
\put(81.5,0){$\cdot$}
\put(81.5,-3){$3$}
\put(86.5,32){$\vdots$}
\put(86.5,29){$\vdots$}
\put(86.5,26){$\vdots$}
\put(86.5,23){$\vdots$}
\put(86.5,20){$\vdots$}
\put(86.5,17){$\vdots$}
\put(86.5,14){$\vdots$}
\put(86.5,11){$\vdots$}
\put(86.5,9){$\vdots$}
\put(86.5,6){$\vdots$}
\put(86.5,3){$\vdots$}
\put(86.5,1.5){$\cdot$}
\put(86.5,0){$\cdot$}
\put(86.5,-3){$\frac{7}{2}$}
%\put(53.8,1.65){\line(5,6){3}}
\put(53.8,1.7){$\cdot$}\put(54.3,2.2){$\cdot$}\put(54.8,2.7){$\cdot$}\put(55.3,3.3){$\cdot$}\put(55.8,3.8){$\cdot$}
\put(56.3,4.3){$\cdot$}
%\put(58.75,6.75){\line(5,6){3}}
\put(58.8,6.8){$\cdot$}\put(59.3,7.3){$\cdot$}\put(59.8,7.8){$\cdot$}\put(60.3,8.3){$\cdot$}\put(60.8,8.8){$\cdot$}
\put(61.3,9.3){$\cdot$}
%\put(63.75,11.75){\line(5,6){3}}
\put(63.8,11.8){$\cdot$}\put(64.3,12.3){$\cdot$}\put(64.8,12.8){$\cdot$}\put(65.3,13.3){$\cdot$}\put(65.8,13.8){$\cdot$}
\put(66.3,14.3){$\cdot$}
%\put(68.75,16.75){\line(5,6){3}}
\put(68.8,16.8){$\cdot$}\put(69.3,17.3){$\cdot$}\put(69.8,17.8){$\cdot$}\put(70.3,18.3){$\cdot$}\put(70.8,18.8){$\cdot$}
\put(71.3,19.3){$\cdot$}
%\put(33.75,21.75){\line(5,6){3}}
\put(33.8,21.8){$\cdot$}\put(34.3,22.3){$\cdot$}\put(34.8,22.8){$\cdot$}\put(35.3,23.3){$\cdot$}\put(35.8,23.8){$\cdot$}
\put(36.3,24.3){$\cdot$}
%\put(38.75,26.75){\line(5,6){3}}
\put(38.8,26.8){$\cdot$}\put(39.3,27.3){$\cdot$}\put(39.8,27.8){$\cdot$}\put(40.3,28.3){$\cdot$}\put(40.8,28.8){$\cdot$}
\put(41.3,29.3){$\cdot$}
%\put(43.75,31.75){\line(5,6){3}}
\put(43.8,31.8){$\cdot$}\put(44.3,32.3){$\cdot$}\put(44.8,32.8){$\cdot$}\put(45.3,33.3){$\cdot$}\put(45.8,33.8){$\cdot$}
\put(46.3,34.3){$\cdot$}
%\put(48.75,36.75){\line(5,6){3}}
\put(48.8,36.8){$\cdot$}\put(49.3,37.3){$\cdot$}\put(49.8,37.8){$\cdot$}\put(50.3,38.3){$\cdot$}\put(50.8,38.8){$\cdot$}
\put(51.3,39.3){$\cdot$}
%\put(47.25,5.25){\line(5,-6){3}}
\put(47.3,4.4){$\cdot$}\put(47.8,3.9){$\cdot$}\put(48.3,3.4){$\cdot$}\put(48.8,2.9){$\cdot$}\put(49.3,2.4){$\cdot$}
\put(49.8,1.9){$\cdot$}
%\put(42.25,10.25){\line(5,-6){3}}
\put(42.3,9.4){$\cdot$}\put(42.8,8.9){$\cdot$}\put(43.3,8.4){$\cdot$}\put(43.8,7.9){$\cdot$}\put(44.3,7.4){$\cdot$}
\put(44.8,6.9){$\cdot$}
%\put(37.25,15.25){\line(5,-6){3}}
\put(37.3,14.4){$\cdot$}\put(37.8,13.9){$\cdot$}\put(38.3,13.4){$\cdot$}\put(38.8,12.9){$\cdot$}\put(39.3,12.4){$\cdot$}
\put(39.8,11.9){$\cdot$}
%\put(32.25,20.25){\line(5,-6){3}}
\put(32.3,19.4){$\cdot$}\put(32.8,18.9){$\cdot$}\put(33.3,18.4){$\cdot$}\put(33.8,17.9){$\cdot$}\put(34.3,17.4){$\cdot$}
\put(34.8,16.9){$\cdot$}
%\put(67.25,25.25){\line(5,-6){3}}
\put(67.3,24.4){$\cdot$}\put(67.8,23.9){$\cdot$}\put(68.3,23.4){$\cdot$}\put(68.8,22.9){$\cdot$}\put(69.3,22.4){$\cdot$}
\put(69.8,21.9){$\cdot$}
%\put(62.25,30.25){\line(5,-6){3}}
\put(62.3,29.4){$\cdot$}\put(62.8,28.9){$\cdot$}\put(63.3,28.4){$\cdot$}\put(63.8,27.9){$\cdot$}\put(64.3,27.4){$\cdot$}
\put(64.8,26.9){$\cdot$}
%\put(57.25,35.25){\line(5,-6){3}}
\put(57.3,34.4){$\cdot$}\put(57.8,33.9){$\cdot$}\put(58.3,33.4){$\cdot$}\put(58.8,32.9){$\cdot$}\put(59.3,32.4){$\cdot$}
\put(59.8,31.9){$\cdot$}
%\put(52.25,40.25){\line(5,-6){3}}
\put(52.3,39.4){$\cdot$}\put(52.8,38.9){$\cdot$}\put(53.3,38.4){$\cdot$}\put(53.8,37.9){$\cdot$}\put(54.3,37.4){$\cdot$}
\put(54.8,36.9){$\cdot$}
\end{picture}
\end{center}
\vspace{0.3cm}
\begin{center}\begin{minipage}{30pc}{\small {\bf Fig.\,1:} Eigenvector subspaces $\sH_E\simeq\Sym_{(k,r)}$ of the energy operator $H$. The each subspace $\sH_E$ (level of matter spectrum) is a space of irreducible representation $\boldsymbol{\tau}_{k/2,r/2}$ belonging to a system of interlocking representations of the Lorentz group. The first cell of spinorial chessboard of second order is marked by dotted lines.}\end{minipage}\end{center}
\end{figure}

Further, in virtue of the isomorphism $\SL(2,\C)\simeq\spin_+(1,3)$ we will consider the universal covering $\widetilde{G}$ as a \textit{spinor group}. It allows us to associate in addition a \textit{spinor structure} with the each cyclic vector $\left|\Phi\right\rangle\in\sH_\infty$ (in some sense, it be a second layer in `clothing' of the operator algebra). Spintensor representations of the group $\widetilde{G}\simeq\spin_+(1,3)$ form a \textit{substrate} of interlocking representations $\boldsymbol{\tau}_{k/2,r/2}$ of the Lorentz group realized in the spaces $\Sym_{(k,r)}\subset\dS_{2^{k+r}}$, where $\dS_{2^{k+r}}$ is a spinspace. In its turn, as it is known \cite{Lou91}, a spinspace is a minimal left ideal of the Clifford algebra $\cl_{p,q}$, that is, there exists an isomorphism $\dS_{2^m}(\K)\simeq I_{p,q}=\cl_{p,q}f$, where $f$ is a primitive idempotent of $\cl_{p,q}$, and $\K=f\cl_{p,q}f$ is a division ring of the algebra $\cl_{p,q}$, $m=(p+q)/2$. The complex spinspace $\dS_{2^m}(\C)$ is a complexification $\C\otimes I_{p,q}$ of the minimal left ideal $I_{p,q}$ of the real subalgebra $\cl_{p,q}$. So, $\dS_{2^{k+r}}(\C)$ is the minimal left ideal of the complex algebra $\C_{2k}\otimes\overset{\ast}{\C}_{2r}\simeq\C_{2(k+r)}$ (for more details see \cite{Var15d,Var16}). Let us define a system of \textit{basic cyclic vectors} endowed with the complex spinor structure (these vectors correspond to the system of interlocking representations of the Lorentz group):
\[
\scriptstyle
\mid\C_0,\boldsymbol{\tau}_{0,0}(H)\Phi\rangle;
\]
\[
\scriptstyle
\mid\C_2,\boldsymbol{\tau}_{1/2,0}(H)\Phi\rangle,\quad\mid\overset{\ast}{\C}_2,\boldsymbol{\tau}_{0,1/2}(H)\Phi\rangle;
\]
\[
\scriptstyle
\mid\C_2\otimes\C_2,\boldsymbol{\tau}_{1,0}(H)\Phi\rangle,\quad
\mid\C_2\otimes\overset{\ast}{\C}_2,\boldsymbol{\tau}_{1/2,1/2}(H)\Phi\rangle,\quad
\mid\overset{\ast}{\C}_2\otimes\overset{\ast}{\C}_2,\boldsymbol{\tau}_{0,1}(H)\Phi\rangle;
\]
\[
\scriptstyle
\mid\C_2\otimes\C_2\otimes\C_2,\boldsymbol{\tau}_{3/2,0}(H)\Phi\rangle,\quad
\mid\C_2\otimes\C_2\otimes\overset{\ast}{\C}_2,\boldsymbol{\tau}_{1,1/2}(H)\Phi\rangle,\quad
\mid\C_2\otimes\overset{\ast}{\C}_2\otimes\overset{\ast}{\C}_2,\boldsymbol{\tau}_{1/2,1}(H)\Phi\rangle,\quad
\mid\overset{\ast}{\C}_2\otimes\overset{\ast}{\C}_2\otimes\overset{\ast}{\C}_2,\boldsymbol{\tau}_{0,3/2}(H)\Phi\rangle;
\]
\[
\ldots\ldots\ldots\ldots\ldots\ldots\ldots\ldots\ldots\ldots\ldots\ldots\ldots\ldots\ldots\ldots\ldots\ldots\ldots
\]
Therefore, in accordance with GNS construction we have complex vector states of the form
\begin{equation}\label{VectState3}
\omega^c_\Phi(H)=
\frac{\langle\Phi\mid\C_{2(k+r)},\boldsymbol{\tau}_{k/2,r/2}(H)\Phi\rangle}{\langle\Phi\mid\Phi\rangle},
\end{equation}
The states $\omega^c_\Phi(H)$ are associated with the complex representations $\boldsymbol{\tau}_{k/2,r/2}(H)$ and cyclic vectors $\left|\Phi\right\rangle\in\sH_\infty$.

As is known, in the Lagrangian formalism of the standard (local) quantum field theory \textit{charged particles} are described by \textit{complex fields}. In our case, pure states of the form (\ref{VectState3}) correspond to \textit{charged states}. At this point, the sign of charge is changed under action of the pseudoautomorphism $\cA\rightarrow\overline{\cA}$ of the complex spinor structure (for more details see \cite{Var01a,Var04,Var14}). Following to analogy with the Lagrangian formalism, where \textit{neutral particles} are described by \textit{real fields}, we introduce vector states of the form
\begin{equation}\label{VectState4}
\omega^r_\Phi(H)=
\frac{\langle\Phi\mid\cl_{p,q},\boldsymbol{\tau}_{k/2,r/2}(H)\Phi\rangle}{\langle\Phi\mid\Phi\rangle}.
\end{equation}
The states (\ref{VectState4}) are associated with the real representations $\boldsymbol{\tau}_{k/2,r/2}(H)$, that is, these representations are endowed with a \textit{real spinor structure}, where $\cl_{p,q}$ is a real subalgebra of $\C_{2(k+r)}$. States of the form (\ref{VectState4}) correspond to \textit{neutral states}. Since the real spinor structure is appeared in the result of reduction $\C_{2(k+r)}\rightarrow\cl_{p,q}$, then (as a consequence) a \textit{charge conjugation} $C$ (pseudoautomorphism $\cA\rightarrow\overline{\cA}$) for the algebras $\cl_{p,q}$ over the real number field $\F=\R$ and quaternionic division ring $\K\simeq\BH$ (the types $p-q\equiv 4,6\pmod{8}$) is reduced to \textit{particle-antiparticle interchange} $C^\prime$ (see \cite{Var01a,Var04,Var14}). As is known, there exist two classes of neutral particles: 1) particles which have antiparticles, such as neutrons, neutrino\footnote{However, it should be noted that the question whether neutrinos are Dirac or Majorana particles (truly neutral fermions) is still open (the last hypothesis being preferred by particle physicists).} and so on; 2) particles which coincide with their antiparticles (for example, photons, $\pi^0$-mesons and so on), that is, so-called \textit{truly neutral particles}. The first class is described by neutral states $\omega^r_\Phi(H)$ with the algebras $\cl_{p,q}$ over the field $\F=\R$ with the rings $\K\simeq\BH$ and $\K\simeq\BH\oplus\BH$ (types $p-q\equiv 4,6\pmod{8}$ and $p-q\equiv 5\pmod{8}$). With the aim to describe the second class of neutral particles we introduce \textit{truly neutral states} $\omega^{r_0}_\Phi(H)$ with the algebras $\cl_{p,q}$ over the number field $\F=\R$ and real division rings $\K\simeq\R$ and $\K\simeq\R\oplus\R$ (types $p-q\equiv 0,2\pmod{8}$ and $p-q\equiv 1\pmod{8}$). In the case of states $\omega^{r_0}_\Phi(H)$ pseudoautomorphism $\cA\rightarrow\overline{\cA}$ is reduced to identical transformation (particle coincides with its antiparticle).

Further, if $\rho$ is the density matrix in $\sH_\infty$, then $\omega^{c}_\rho(H)$, $\omega^{r}_\rho(H)$ and $\omega^{r_0}_\rho(H)$ are statistical mixtures of charged, neutral and truly neutral states.

Before we proceed with a construction of physical Hilbert space, let us consider in more detail the structure of states (\ref{VectState2})-(\ref{VectState4}). The states (\ref{VectState2})-(\ref{VectState4}) present the levels of matter spectrum, that is, actualized (localized) states of quantum micro-objects (`elementary particles').  The each state of the form (\ref{VectState2})-(\ref{VectState4}) possesses the following characteristics (properties): energy (mass), spin and charge (the first two layers in `clothing' of the operator algebra). On this level of description a \textit{state} acquires a primary meaning (in spirit of Birkhoff-von Neumann-Mackey intepretation), and \textit{observable} characteristics (energy, spin, charge, $\ldots$) are properties of the state. Subsequent `clothing' of the operator algebra leads to introduction of new properties (characteristics) of the state. For example, discrete symmetries (space inversion $P$, time reversal $T$, charge conjugation $C$ and their combinations) are appeared as automorphisms of spinor structure associated with the each state of the form (\ref{VectState2})-(\ref{VectState4}) \cite{Var01,Var05c,Var05,Var11} (see also recent paper on spinors transformations \cite{Bud16}). Further, a fractal structure of matter spectrum (third layer of `clothing') is defined by the Cartan-Bott periodicity of spinor structure \cite{Var12}. However, detailed consideration of these characteristics of the states be beyond the scope of the present paper.

\subsection{Physical Hilbert space}
A set of pure states $\omega_\Phi(H)$, defined according to GNS construction by the equality (\ref{VectState2}), at the execution of condition $\omega_\Phi(H)\geq 0$ forms a \textit{physical Hilbert space}
\[
\bsH_{\rm phys}=\bsH^S\otimes\bsH^Q\otimes\bsH_\infty.
\]
It is easy to verify that axioms of addition, multiplication and scalar (inner) product are fulfilled for the vectors $\omega_\Phi(H)\rightarrow\left|\Psi\right\rangle\in\bsH_{\rm phys}$. We assume that a so-defined Hilbert space is \textit{nonseparable}, that is, in general case the axiom of separability is not executed in $\bsH_{\rm phys}$.

The space $\bsH_{\rm phys}$ is a second member of the pair $(\sH_\infty,\bsH_{\rm phys})$, which defines two-level structure of the Hilbert space of elementary particle. Therefore, $\bsH_{\rm phys}$ describes spin and charge degrees of freedom of the particle. In accordance with the charge degrees of freedom we separate three \textit{basic subspaces} in $\bsH_{\rm phys}$.\\
1) \textit{Subspace of charged states} $\bsH^\pm_{\rm phys}=\bsH^S\otimes\bsH^\pm\otimes\bsH_\infty$.\\
2) \textit{Subspace of neutral states} $\bsH^0_{\rm phys}=\bsH^S\otimes\bsH^0\otimes\bsH_\infty$.\\
3) \textit{Subspace of truly neutral states} $\bsH^{\overline{0}}_{\rm phys}=\bsH^S\otimes\bsH^{\overline{0}}\otimes\bsH_\infty$.\\
Basis vectors $\left|\Psi\right\rangle\in\bsH^\pm_{\rm phys}$ are formed by the states $\omega^c_\Phi(H)$ (see (\ref{VectState3})). Correspondingly, $\left|\Psi\right\rangle\in\bsH^0_{\rm phys}$ and $\left|\Psi\right\rangle\in\bsH^{\overline{0}}_{\rm phys}$ are formed by the states $\omega^{r}_\rho(H)$ and $\omega^{r_0}_\rho(H)$.

Following to Birkhoff-von Neumann-Mackey interpretation \cite{BJ36,Mac57}, we assume that on the level of physical Hilbert space  $\bsH_{\rm phys}$ the original (primary) notion is a \textit{state} (\textbf{\textit{ray}}). Let $\left|\Psi\right\rangle$ be a state vector in the space $\bsH_{\rm phys}$, then $\boldsymbol{\Psi}=e^{i\alpha}\left|\Psi\right\rangle$, where $\alpha$ runs all the real numbers and $\sqrt{\left\langle\Psi\right.\!\left|\Psi\right\rangle}=1$, is called a \emph{unit ray}. Therefore, the unit ray $\boldsymbol{\Psi}$ is a totality of basis state vectors $\{\lambda\left|\Psi\right\rangle\}$, $\lambda=e^{i\alpha}$, $\left|\Psi\right\rangle\in\bsH_{\rm phys}$. As is known, the magnitudes, related with observable effects, are absolute values of a semibilinear form $|\left\langle\Psi_1\right.\!\left|\Psi_2\right\rangle|^2$ (these values do not depend on the parameters $\lambda$ characterizing the ray). Thus, a \textit{ray space} is a quotient space $\hat{H}=\bsH_{\rm phys}/S^1$, that is, a projective space of one-dimensional subspaces from $\bsH_{\rm phys}$.
All the states of physical (quantum) system (in our case, elementary particle) are described by the unit rays. We assume that basic correspondence between physical states and elements (rays) of the space $\bsH_{\rm phys}$ includes a \textit{superposition principle} of quantum theory, that is, there exists a collection of basis states such that arbitrary states can be constructed from them via the linear superpositions. Hence it follows a definition of elementary particle. \textit{An elementary particle (\textbf{single quantum microsystem}) is a superposition of state vectors in physical Hilbert space} $\bsH_{\rm phys}$.

\subsubsection{Group action in $\bsH_{\rm phys}$}
We assume that one and the same quantum system can be described by the two different ways in one and the same subspace $\bsH^\pm_{\rm phys}$ ($\bsH^0_{\rm phys}$ or $\bsH^{\overline{0}}_{\rm phys}$) of the space $\bsH_{\rm phys}$ one time by the rays $\boldsymbol{\Psi}_1$, $\boldsymbol{\Psi}_2$, $\ldots$ and other time by the rays $\boldsymbol{\Psi}^\prime_1$, $\boldsymbol{\Psi}^\prime_2$, $\ldots$. One can say that we have here a symmetry of the quantum system when one and the same physical state is described with the help of $\boldsymbol{\Psi}_1$ in the first case and with the help of $\boldsymbol{\Psi}^\prime_1$ in the second case such that probabilities of transitions are the same. Therefore, we have a mapping $\hat{T}$ between the rays $\boldsymbol{\Psi}_1$ and $\boldsymbol{\Psi}^\prime_1$. Since only the absolute values are invariant, then the transformation $\hat{T}$ in the Hilbert space $\bsH_{\rm phys}$ should be unitary or antiunitary. Both these possibilities are realized in the case of subspace $\bsH^\pm_{\rm phys}$, state vectors of which are endowed with the complex spinor structure, because the complex field has two (and only two) automorphisms preserving absolute values: an identical automorphism and complex conjugation. Also both these possibilities are realized in the subspace $\bsH^0_{\rm phys}$, since in this case state vectors are endowed with the real spinor structure with the quaternionic division ring. In the case of subspace $\bsH^{\overline{0}}_{\rm phys}$ we have only unitary transformations $\hat{T}$, because the real spinor structure with the real division ring admits only one identical automorphism.

Let $\left|\Psi_1\right\rangle$, $\left|\Psi_2\right\rangle$, $\ldots$ be the unit vectors chosen from the first totality of rays $\boldsymbol{\Psi}_1$, $\boldsymbol{\Psi}_2$, $\ldots$ and let $\left|\Psi^\prime_1\right\rangle$, $\left|\Psi^\prime_2\right\rangle$, $\ldots$ be the unit vectors chosen from the second totality $\boldsymbol{\Psi}^\prime_1$, $\boldsymbol{\Psi}^\prime_2$, $\ldots$ such that a correspondence $\left|\Psi_1\right\rangle\leftrightarrow\left|\Psi^\prime_1\right\rangle$, $\left|\Psi_2\right\rangle\leftrightarrow\left|\Psi^\prime_2\right\rangle$, $\ldots$ is unitary or antiunitary. The first collection corresponds to the states $\{\omega\}$, and the second collection corresponds to transformed states $\{g\omega\}$. We choose the vectors $\left|\Psi_1\right\rangle\in\boldsymbol{\Psi}_1$, $\left|\Psi_2\right\rangle\in\boldsymbol{\Psi}_2$, $\ldots$ and $\left|\Psi^\prime_1\right\rangle\in\boldsymbol{\Psi}^\prime_1$, $\left|\Psi^\prime_2\right\rangle\in\boldsymbol{\Psi}^\prime_2$, $\ldots$ such that
\begin{equation}\label{Action}
\left|\Psi^\prime_1\right\rangle=T_g\left|\Psi_1\right\rangle,\quad
\left|\Psi^\prime_2\right\rangle=T_g\left|\Psi_2\right\rangle,\quad\ldots
\end{equation}
It means that if $\left|\Psi_1\right\rangle$ is the vector associated with the ray $\boldsymbol{\Psi}_1$, then $T_g\left|\Psi_1\right\rangle$ is the vector associated with the ray $\boldsymbol{\Psi}^\prime_1$. If there exist two operators $T_g$ and $T_{g^\prime}$ with the property (\ref{Action}), then they can be distinguished by only a constant factor. Therefore,
\begin{equation}\label{Product}
T_{gg^\prime}=\phi(g,g^\prime)T_gT_{g^\prime},
\end{equation}
where $\phi(g,g^\prime)$ is a phase factor. Representations of the type (\ref{Product}) are called \emph{ray (projective) representations}. It means also that we have here a correspondence between physical states and rays of the Hilbert space $\bsH_{\rm phys}$. Hence it follows that the ray representation $T$ of a \emph{topological group} $G$ is a continuous homomorphism $T:\;G\rightarrow L(\hat{H})$, where $L(\hat{H})$ is a set of linear operators in the projective space $\hat{H}$ endowed with a factor-topology according to the mapping $\hat{H}\rightarrow\bsH_{\rm phys}$, that is, $\left|\Psi\right\rangle\rightarrow\boldsymbol{\Psi}$. However, when $\phi(g,g^\prime)\neq 1$ we cannot to apply the mathematical theory of usual group representations. With the aim to avoid this obstacle we construct a more large group $\bcE$ in such manner that usual representations of $\bcE$ give all nonequivalent ray representations (\ref{Product}) of the group $G$. This problem can be solved by the \emph{lifting} of projective representations of $G$ to usual representations of the group $\bcE$. Let $\bcK$ be an Abelian group generated by the multiplication of nonequivalent phases $\phi(g,g^\prime)$ satisfying the condition
\[
\phi(g,g^\prime)\phi(gg^\prime,g^{\prime\prime})=\phi(g^\prime,g^{\prime\prime})\phi(g,g^\prime g^{\prime\prime}).
\]
Let us consider the pairs $(\phi,x)$, where $\phi\in\bcK$, $x\in G$, in particular, $\bcK=\{(\phi,e)\}$, $G=\{(e,x)\}$. The pairs $(\phi,x)$ form a group with multiplication law of simidirect product type:  $(\phi_1,x_1)(\phi_2,x_2)=(\phi_1\phi(x_1,x_2)\phi_2,x_1x_2)$. The group $\bcE=\{(\phi,x)\}$ is called a \emph{central extension} of the group $G$ via the group $\bcK$. Vector representations of the group $\bcE$ contain all the ray representations of the group $G$. Hence it follows that a symmetry group $G$ of physical system induces a unitary or antiunitary representation $T$ of invertible mappings of the space $\bsH^S\otimes\bsH^Q\otimes\bsH_\infty$ into itself, which is a representation of the central extension $\bcE$ of $G$.

On the level of physical Hilbert space $\bsH_{\rm phys}$ a symmetry group $G$ is understood as one from the sequence of unitary unimodular groups: $\SU_T(2)$ (isospin group), $\SU(3)$, $\ldots$, $\SU(N)$, $\ldots$ (groups of so-called `internal' symmetries). According to Heisenberg, the groups $\SU(N)$ define \textit{dynamical (\textbf{secondary}) symmetries}. Thus, in conformity with the two-level structure of the Hilbert space of elementary particle (single quantum microsystem), defined by the pair $(\sH_\infty,\bsH_{\rm phys})$, all the set of symmetry groups $G$ is divided on the two classes: 1) \textit{groups of fundamental (primary) symmetries} $G_f$ which form state vectors of quantum microsystem; 2) \textit{groups of dynamical (secondary) symmetries} $G_d$ which describe approximate symmetries (transitions) between state vectors of quantum system\footnote{According to Wigner \cite{Wig39}, a quantum system, described by an irreducible unitary \textit{representation} of the Poincar\'{e} group $\cP$, is called an elementary particle. On the other hand, in accordance with $\SU(3)$-theory an elementary particle is described by a \emph{vector} of irreducible representation of the group $\SU(3)$. For example, in a so-called `eightfold way' \cite{GN64} the hadrons (baryons and mesons) are represented by the vectors of eight-dimensional regular representation $\Sym^0_{(1,1)}$ of the group $\SU(3)$. Thus, we have two mutually exclusive each other interpretations of elementary particle: as a \textit{representation} of the group $\cP$ and as a \textit{vector} of the representation of the group $\SU(3)$. This opposition vanishes if we adopt that all the `elementary particles' are localized states (levels) of matter spectrum. At this point, matter spectrum is realized via $\bsH_{\rm phys}$, the vectors of which are defined by cyclic representations of the operator algebra (energy operator $H$). On this level of description the group $\SU(3)$, defined in $\bsH_{\rm phys}$ via central extension, describes dynamical (approximate) symmetries between different states (more precisely, between states from different coherent subspaces in $\bsH_{\rm phys}$).}.

\subsubsection{Reduction principle}
So, dynamical symmetries $G_d$ relate different states (state vectors $\left|\Psi\right\rangle\in\bsH_{\rm phys}$) of quantum system. Symmetry $G_d$ can be represented as a \textit{quantum transition} between the states of quantum system (levels of matter spectrum). For example, if we take baryons, then the reaction $N\rightarrow Pe^-\nu$ (neutron decay) can be considered as a transition $N\rightarrow P$, the reaction $\Sigma^-\rightarrow N\pi^-$ as a transition $\Sigma\rightarrow N$ and so on (here we neglect the leptons and mesons). At this point, we assume that representing operators of the complex envelope of Lie algebra of $G_d$ realize all possible quantum transitions of the system. For example, if $G_d=\SU(3)$, then Okubo operators $\boldsymbol{\sf A}^\sigma_\tau$ transit the states (particles) of the octet into each other. It is natural to regard that operators of the group $G_d$ or its subgroup connect allied states. The chain of nested subgroups leads to a hierarchical classification of the states.

A dynamical symmetry is defined by the chain of nested Lie groups:
\[
G=G_0\supset G_1\supset G_2\supset\ldots\supset G_k.
\]
A \textit{system} with the given dynamical symmetry is defined by an irreducible representation $\fP$ of the group $G$ in the space $\bsH_{\rm phys}$. A reduction $G/G_1$ of the representation $\fP$ of the group $G$ on its subgroup $G_1$ leads to a decomposition of $\fP$ into orthogonal sum of irreducible representations $\fP^{(1)}_i$ of the subgroup $G_1$:
\[
\fP=\fP^{(1)}_1\oplus\fP^{(1)}_2\oplus\ldots\oplus\fP^{(1)}_i\oplus\ldots.
\]
In its turn, a reduction $G_1/G_2$ of the representation of the group $G_1$ on its subgroup $G_2$ leads to a decomposition of the representations $\fP^{(1)}_i$ into irreducible representations $\fP^{(2)}_{ij}$ of the group $G_2$:
\[
\fP^{(1)}_i=\fP^{(2)}_{i1}\oplus\fP^{(2)}_{i2}\oplus\ldots\oplus\fP^{(2)}_{ij}\oplus\ldots
\]
and so on\footnote{For example, one of the basic supermultiplets of $\SU(3)$-theory (baryon octet $F_{1/2}$), based on the eight-dimensional regular representation $\Sym^0_{(1,1)}$ of $\SU(3)$, admits the following $\SU(3)/\SU(2)$-reduction into isotopic multiplets of the subgroup $\SU(2)$: $\Sym^0_{(1,1)}=\Phi_3\oplus\Phi_2\oplus\overset{\ast}{\Phi}_2\oplus\Phi_0$, where $\Phi_3$ is a triplet, $\Phi_2$ and
$\overset{\ast}{\Phi}_2$ are doublets, $\Phi_0$ is a singlet. Analogously, for the hypermultiplets of $\SU(6)$-theory (baryon 56-plet and meson 35-plet) there are $\SU(6)/\SU(3)$- and $\SU(6)/\SU(4)$-reductions, where $\SU(4)$ is a Wigner subgroup \cite{RF70,Fet}.}.

\section{Coherent subspaces}
As is known \cite{WWW52}, there are unit rays which are physically unrealizable. There exist physical restrictions (\textit{superselection rules}) on the execution of superposition principle (for more details see \cite{SH701,SH702,Hor75}). In 1952, Wigner, Wightman and Wick \cite{WWW52} showed that existence of superselection rules is related with the measurability of relative phase of the superposition. It means that a pure state cannot be realized in the form of superposition of some states, for example, there is no a pure state (coherent superposition) consisting of boson $\left|\Psi_b\right\rangle$ and fermion $\left|\Psi_f\right\rangle$ states (superselection rule on spin). However, if we define the density matrix $\rho$ in $\bsH_{\rm phys}$, then a superposition $\left|\Psi_b\right\rangle+\left|\Psi_f\right\rangle$ defines a mixed state.
\begin{thm}
Physical Hilbert space $\bsH_{\rm phys}$ is decomposed into a direct sum of (non-null) coherent subspaces
\begin{equation}\label{Decomp1}
\bsH_{\rm phys}=\bsH^\pm_{\rm phys}\bigoplus\bsH^0_{\rm phys}\bigoplus\bsH^{\overline{0}}_{\rm phys},
\end{equation}
where
\begin{equation}\label{Decomp2}
\bsH^Q_{\rm phys}=\bigoplus^{|l-\dot{l}|}_{s=-|l-\dot{l}|}\bsH^{2|s|+1}\otimes\bsH^Q\otimes\bsH_\infty,\quad
Q=\{\pm,0,\overline{0}\}.
\end{equation}
At this point, superposition principle takes place in the restricted form, that is, within the limits of coherent subspaces. A non-null linear combination of vectors of pure states is a vector of pure state at the condition that all original vectors lie in one and the same coherent subspace. A superposition of vectors of pure states from different coherent subspaces defines a mixed state.
\end{thm}
\begin{proof}
An original point of the proof is a correspondence $\omega_\Phi(H)\leftrightarrow\left|\Psi\right\rangle$ between states of operator algebra and basis vectors of the space $\bsH_{\rm phys}$. As it has been shown in the section 2.1, the set of all pure states $PS(\fA)$ of the operator algebra $\fA$ is divided in pairs on disjoint and mutually orthogonal sectors in virtue of the equivalence relation $\omega_1\sim\omega_2$. Sectors coincide with equivalence classes in $PS(\fA)$ (in essence, sector is an algebraic counterpart of the coherent subspace). Further, we assume that a some set of vectors in $\bsH_{\rm phys}$, containing pure states of the algebra $\fA$, form a total set in $\bsH_{\rm phys}$, that is, such a set $X$, closing of linear envelope of $X$ is dense everywhere in $\bsH_{\rm phys}$. Then $X$ cannot be represented as a union of the two (or more) nonempty mutually orthogonal subsets. We adopt that vectors $\left|\Psi_1\right\rangle$, $\left|\Psi_2\right\rangle\in X$ are related by a correspondence $\left|\Psi_1\right\rangle\sim\left|\Psi_2\right\rangle$, if $\left|\Psi_1\right\rangle$ and $\left|\Psi_2\right\rangle$ belong to a linear envelope from $X$. It easy to see that the correspondence $\left|\Psi_1\right\rangle\sim\left|\Psi_2\right\rangle$ is induced by the equivalence relation $\omega_1\sim\omega_2$ from $PS(\fA)$. Therefore, $\left|\Psi_1\right\rangle\sim\left|\Psi_2\right\rangle$ is an equivalence relation and equivalence classes in $X$ form a partition of $X$ onto mutually orthogonal systems $X_\nu$, where $\{\nu\}=N$ is a some index collection. Taking as $\bsH^\nu_{\rm phys}$ a closed linear envelope of the set $X_\nu$, we come to a sought decomposition of $\bsH_{\rm phys}$ into a direct sum of mutually orthogonal subspaces $\bsH^\nu_{\rm phys}$:
\[
\bsH_{\rm phys}=\bigoplus_{\nu\in N}\bsH^\nu_{\rm phys}.
\]
Thus, there is a one-to-one correspondence between pure states and unit rays in $\cup_\nu\bsH^\nu_{\rm phys}$. Hence it follows a restricted form of the superselection principle (namely, within the limits of subspaces $\bsH^\nu_{\rm phys}$). It is obvious that three base subspaces $\bsH^\pm_{\rm phys}$, $\bsH^0_{\rm phys}$ and $\bsH^{\overline{0}}_{\rm phys}$ are coherent subspaces (with respect to charge) of the original physical space $\bsH_{\rm phys}$. Hence it follows the formula (\ref{Decomp1}). A subsequent decomposition of $\bsH_{\rm phys}$ onto coherent subspaces is realized with respect to spin, that is, with respect to a forming component $\bsH^S$ in $\bsH^S\otimes\bsH^Q\otimes\bsH_\infty$ (the formula (\ref{Decomp2})).
\end{proof}
\textbf{Example}. According to modern data, neutrino is a superposition of three mass neutrino states: electron, muon and $\tau$-lepton neutrinos. All the three states lie in coherent subspases $\bsH^2\otimes\bsH^0\otimes\bsH_\infty$, which belong to spin-1/2 line of the base subspace $\bsH^0_{\rm phys}$ (subspace of neutral states). Nevertheless, all three neutrino states of matter spectrum, belonging to one and the same coherent subspace, are differed from each other by the energy value. The difference in energy is characterized by the different arrangement of corresponded representations on the spin-1/2 line.
\subsection{Gauge symmetries}
Let $G=U(1)^n\equiv U(1)\times\ldots U(1)$ be a compact $n$-parameter Abelian group (\textit{gauge group}) defined in $\bsH_{\rm phys}$ via the central extension. An arbitrary element of this group is represented by a collection of $n$ phase factors:
\[
g(s_1,\ldots,s_n)\equiv(e^{i\alpha_1},\ldots,e^{i\alpha_n}),\quad 0\leq\alpha_j<2\pi.
\]
Let us define in the space $\bsH_{\rm phys}$ an exact unitary representation $\fU$ of the group $G$. Gauge transformations in $\bsH_{\rm phys}$ have the form
\[
\fU(g)=s^{Q_1}_1\ldots s^{Q_n}_n\equiv e^{i(\alpha_1Q_1+\ldots+\alpha_nQ_n)}.
\]
Generators $Q_1$, $\ldots$, $Q_n$ of the gauge transformations are mutually commuting self-conjugated operators with integer spectrum. $Q_j$ ($j=1,\ldots,n$) are called \textit{charges}, which correspond to a given gauge group. Then $\bsH_{\rm phys}$ is decomposed into a direct sum
\begin{equation}\label{SRules}
\bsH_{\rm phys}=\bigoplus_{q_1,\ldots,q_n\in\dZ}\bsH_{\rm phys}(q_1,\ldots,q_n)
\end{equation}
of corresponded spectral subspaces consisting of the all vectors $\left|\Psi\right\rangle$ such that $(Q_j-q_j)\left|\Psi\right\rangle=0$. At this point, an arbitrary non-null vector $\left|\Psi\right\rangle\in\bsH_{\rm phys}$ defines a pure state of the algebra $\pi(H)$ exactly when $\left|\Psi\right\rangle$ is an eigenvector for the all charges. Thus, we have \textit{standard (discrete) superselection rules} in $\bsH_{\rm phys}$, and (\ref{SRules}) is a decomposition of $\bsH_{\rm phys}$ into a direct sum of coherent subspaces $\bsH_{\rm phys}(q_1,\ldots,q_n)$. According to modern situation in particle physics, superselection rules can be described completely by electric $Q$ ($=Q_1$), baryon $B$ ($=Q_2$) and lepton $L$ ($=Q_3$) charges, such that a decomposition onto coherent subspaces has the form
\[
\bsH_{\rm phys}=\bigoplus_{q,b,\ell\in\dZ}\bsH_{\rm phys}(q,b,\ell).
\]
The decomposition of $\bsH_{\rm phys}$ onto coherent subspaces with respect to charge and spin is given by the formula (\ref{Decomp1}). With the aim to describe all spectrum of observed states (levels of matter spectrum) we introduce 2-parameter gauge group $G=U(1)^2\equiv U(1)\times U(1)$ with respect to baryon $B$ and lepton $L$ charges. Then the decomposition of $\bsH_{\rm phys}$ onto coherent subspaces takes the following form:
\begin{equation}\label{Decomp3}
\bsH_{\rm phys}=\bigoplus_{b,\ell\in\dZ}\left[\bsH^\pm_{\rm phys}(b,\ell)\bigoplus\bsH^0_{\rm phys}(b,\ell)\bigoplus
\bsH^{\overline{0}}_{\rm phys}(b,\ell)\right],
\end{equation}
where
\[
\bsH^Q_{\rm phys}(b,\ell)=\bigoplus^{|l-\dot{l}|}_{s=-|l-\dot{l}|}\bsH^{2|s|+1}\otimes\bsH^Q(b,\ell)\otimes\bsH_\infty,\quad
Q=\{\pm,0,\overline{0}\}.
\]
The decomposition (\ref{Decomp3}) allows one to embrace practically the all observed spectrum of states (see Particle Data Group). First of all, matter spectrum is divided onto three sectors: lepton, meson and baryon sectors. Lepton sector includes into itself charged leptons: electron $e^-$, muon $\mu^-$, $\tau^-$-lepton (and their antiparticles). All charged leptons belong to coherent subspaces of the form $\bsH^\pm_{\rm phys}(0,\ell)$. Neutral leptons (three kinds of neutrino) belong to coherent subspace $\bsH^0_{\rm phys}(0,\ell)$. Lepton sector includes also one truly neutral state: photon $\gamma$ (subspace $\bsH^{\overline{0}}_{\rm phys}(0,\ell)$). In contrast to lepton sector, meson and baryon sectors (hadron sector in the aggregate) include a wide variety of states (particles). Meson sector is divided (with respect to charge) onto three sets of coherent subspaces. At first, charged mesons ($\pi^\pm$ (pions), $K^\pm$ (kaons), $\rho^\pm$, $\ldots$) belong to coherent subspaces of the form $\bsH^\pm_{\rm phys}(0,0)$ with integer spin (all mesons have integer spin). Further, neutral mesons ($K^0$, $\overset{\ast}{K}{}^0$, $\ldots$) belong to subspaces $\bsH^0_{\rm phys}(0,0)$ of integer spin. In turn, truly neutral mesons ($\pi^0$, $\eta$, $\varphi$, $\rho^0$, $\ldots$) are the states belonging to coherent subspace $\bsH^{\overline{0}}_{\rm phys}(0,0)$. The baryon sector is divided with respect to a charge on the two sets of coherent subspaces: charged baryons ($p$ (proton), $\Sigma^\pm$, $\Xi^\pm$, $\ldots$) form subspaces $\bsH^\pm_{\rm phys}(b,0)$ with half-integer spin (all baryons have half-integer spin); neutral baryons ($n$ (neutron), $\Sigma^0$, $\Xi^0$, $\ldots$) are the states from coherent subspaces $\bsH^0_{\rm phys}(b,0)$ of half-integer spin. Truly neutral baryons are not discovered until now.

In conclusion it should be noted that an addition of gauge symmetries leads to a triple symmetry ($G_f$, $G_d$, $G_g$) division of matter spectrum.  Namely, fundamental symmetries $G_f$ participate in formation of pure states (rays) of quantum system and coherent subspaces in $\bsH_{\rm phys}$, dynamical symmetries $G_d$ describe transitions between states from different coherent subspaces, and gauge symmetries $G_g$ relate pure states within coherent subspaces.

\end{document}